\documentclass[conference,letterpaper]{IEEEtran}
\usepackage{cite}
\usepackage{amsmath,amssymb,amsfonts,amsthm}
\usepackage{ifthen}
\usepackage{tikz}
\usepackage{cite}
\usepackage{url}
\usepackage{titlesec}
\usepackage{subfig}
\usepackage{blkarray}
\usepackage{dsfont}
\usepackage[mathscr]{euscript}
\usepackage{enumitem}
\usepackage{algorithm}
\usepackage[noend]{algpseudocode}
\usepackage{blkarray}
\newtheorem{theorem}{Theorem}[section]

\newtheorem{proposition}{Proposition}[section]
\newtheorem{lemma}{Lemma}[section]
\theoremstyle{remark}

\theoremstyle{definition}
\newtheorem{definition}{Definition}[section]

\usepackage{subcaption}
\usepackage{blkarray}
\usepackage{dsfont}

\usepackage{url}
\usepackage{tcolorbox}
\usepackage{arydshln}
\usepackage{mathtools}
\usepackage{dsfont}

\usepackage{xcolor}

\usepackage{xcolor}

\definecolor{brickred}{cmyk}{0,0.89,0.94,0.28}
\definecolor{goldenrod}{cmyk}{0,0.10,0.84,0}
\definecolor{purple}{cmyk}{0.45,0.86,0,0}
\definecolor{rawsienna}{cmyk}{0,0.72,1,0.45}
\definecolor{olivegreen}{cmyk}{0.64,0,0.95,0.40}
\definecolor{peach}{cmyk}{0,0.5,0.7,0}
\definecolor{darkolive}{rgb}{0.,0.4,0.}
\colorlet{grey}{gray!40}

\global\long\def\E{\mathbb{E}}

\usepackage[cmintegrals]{newtxmath}
\newlist{mylist}{enumerate}{3}
\setlist[mylist]{label={}}
\usepackage{graphicx}
\usepackage{textcomp}
\def\BibTeX{{\rm B\kern-.05em{\sc i\kern-.025em b}\kern-.08em
		T\kern-.1667em\lower.7ex\hbox{E}\kern-.125emX}}

\ifCLASSINFOpdf
\else
\fi

\begin{document}
\title{Bounding User Contributions for User-Level Differentially Private Mean Estimation}
\author{
	\IEEEauthorblockN{V.~Arvind~Rameshwar}
	\IEEEauthorblockA{Department of Electrical Engineering\\ IIT Madras, Chennai, India\\
		Email: \url{arvind@ee.iitm.ac.in}}
	\and
	\IEEEauthorblockN{Anshoo Tandon}
	\IEEEauthorblockA{IUDX Program Unit\\IISc, Bengaluru, India\\
		Email: \url{anshoo.tandon@gmail.com}}
}
\IEEEoverridecommandlockouts
%\thanks{The authors with the India Urban Data Exchange Program Unit, Indian Institute of Science, Bengaluru 560012, India, email: \texttt{arvind.rameshwar@gmail.com, anshoo.tandon@gmail.com}.}% <-this % stops a 
\maketitle

\begin{abstract}
	We revisit the problem of releasing the sample mean of bounded samples in a dataset, privately, under user-level $\varepsilon$-differential privacy (DP). We aim to derive the optimal method of preprocessing data samples, within a canonical class of processing strategies, in terms of the estimation error. Typical error analyses of such \emph{bounding} (or \emph{clipping}) strategies in the literature assume that the data samples are independent and identically distributed (i.i.d.), and sometimes also that all users contribute the same number of samples (data homogeneity)---assumptions that do not accurately model real-world data distributions. Our main result in this work is a precise characterization of the preprocessing strategy that gives rise to the smallest \emph{worst-case} error over all datasets -- a \emph{distribution-independent} error metric -- while allowing for data heterogeneity. We also show via experimental studies that even for i.i.d. real-valued samples, our clipping strategy performs much better, in terms of \emph{average-case} error, than the widely used bounding strategy of Amin et al. (2019).
\end{abstract}
\section{Introduction}
%Federated learning (FL) is now a widely studied framework for the collaborative training of a machine learning (ML) model by a large collection of devices (also called clients or users), by using decentralized training data. Typically, such a training task proceeds iteratively; in every round, each participating client performs a local update to the machine learning model using its training data and then sends this update to a central server, which is tasked with ``aggregating" the received updates. The aggregated model is then communicated back to the clients, which then begin a fresh round of local model updates. A canonical FL algorithm is the \texttt{FedAvg} algorithm introduced in \cite{fl-mcmahan}. Importantly, FL allows for the training of a large-scale ML model in a distributed fashion, while ensuring that the training data remains private to each client. For a thorough survey of the challenges in real-world deployment of FL, we refer the reader to \cite{kairouz-survey,li-survey}.

In this article, we work on the fundamental problem of processing bounded, potentially vector-valued samples in a dataset, for the release of a private estimate of the sample mean. In particular, we work within the framework of ``user-level" differential privacy \cite{userlevel,hetero}, which is a generalization of the now widely adopted framework of differential privacy (DP) \cite{dworkroth, vadhan2017} for the design and analysis of privacy-preserving algorithms. Loosely, user-level DP guarantees the privacy of a ``user", who could contribute more than one sample, by ensuring the statistical indistinguishability of outputs of the algorithm to changes in the user's samples. User-level DP has practical relevance for inference tasks on real-world datasets, such as traffic datasets, datasets of user expenditures, and time series data, where different users contribute potentially different numbers of samples (data heterogeneity) \cite{usereg3,usereg4}, and more recently, in federated learning (FL) applications (see, e.g., \cite[Sec. 4]{kairouz-survey}). 
%Moreover, user-level DP algorithms are increasingly becoming popular subroutines for integration into
%federated learning (FL) frameworks -- see, e.g., \cite[Sec. 4]{kairouz-survey} for more details. 

%User-level privacy assumes significance in the context of real-world IoT datasets, such as traffic databases, which record multiple contributions from every user (see also \cite{usereg1,usereg2,usereg3,usereg4}). Specifically, user-level DP mechanisms for mean estimation \cite{userlevel,usereg4} are employed by the central server in a canonical FL algorithm such as \texttt{FedAvg} \cite{fl-mcmahan}, while computing a weighted average of gradients passed by each client; these mechanisms help ensure the privacy of a single user (or client), who potentially contributes several samples to the training dataset. 

There are two key requirements of such user-level DP mechanisms for mean estimation. Firstly, the mechanisms must be designed to work with heterogeneous data. Secondly, one would like reliable reconstruction of the true sample mean, even when the data samples are non-i.i.d. (independent and identically distributed). Our focus is hence to  characterize an error metric, which is independent of the underlying data distribution and can be explicitly computed and optimized, for heterogeneous data.

%The second requirement is of particular importance in FL applications, where each client would like to reliably reconstruct the aggregated model update, in every round, from the noisy, privatized aggregate released by the central server, while running the FL algorithm on highly correlated real-world datasets.  

In this article, we confine our attention to user-level (pure) $\varepsilon$-DP algorithms for mean estimation. A key subroutine in such mechanisms \cite{userlevel,hetero-user-level,hist-user-level,amin} is the preprocessing of the data samples for the release of an estimate of the sample mean, which requires the addition of less noise for privacy, as against releasing a noised version of the true sample mean. Such a preprocessing procedure (or ``bounding" or ``clipping" strategy \cite{amin}) either drops selected samples, or projects the samples to a ``high-probability subset". While it is usually easy to establish that such mechanisms are differentially private, an analysis of their ``utility", or the error in estimation of the true statistic, often relies on distributional assumptions about the dataset. 

In this work, following \cite{dp_spcom,dp_preprint,tit-preprint}, we define and explicitly compute the \emph{worst-case error}, over all datasets, of general preprocessing (or bounding) strategies. This error metric is natural in settings with arbitrarily correlated data, where a user potentially ascribes his/her error tolerance to the worst dataset that the statistic is computed on. Furthermore, this error metric is distribution-independent and, as we show, is computable under data heterogeneity too. We then explicitly identify the bounding strategy that results in the smallest worst-case error. Interestingly, we also observe from experimental studies that for scalar samples, our clipping strategy also gives rise to much smaller errors \emph{on average} compared to that in \cite{amin}, for selected dataset sizes, when the samples are drawn i.i.d. according to common distributions.

\section{Notation and Preliminaries}
\label{sec:prelim}
\subsection{Notation}
The notation $\mathbb{N}$ denotes the set of positive natural numbers. For $n\in \mathbb{N}$, the notation $[n]$ denotes the set $\{1,2,\ldots,n\}$. 
%An empty product is defined to be $1$. We use the notation $\mathbf{b}$ to denote a vector of predefined length, all of whose symbols are equal to $b\in \mathbb{R}$. 
%Given a real-valued vector $\mathbf{a} = (a_1,a_2,\ldots,a_n)$, we denote its $\ell_1$-norm by $\lVert \mathbf{a} \rVert_1:= \sum_{i=1}^n |a_i|$ and its $\ell_2$-norm by $\lVert \mathbf{a} \rVert_2:= \left(\sum_{i=1}^n a_i^2\right)^{1/2}$. 
%Random variables are denoted by upper-case letters, e.g., $X, Y, \mathbf{Z}$.
%Given a collection of real numbers $(x_1,\ldots,x_n)$, we define its median med$(x_1,\ldots,x_n)$ as any value $x$ such that $|\{i\in [n]: x_i\geq x\}| = \left \lceil n/2\right \rceil$. 
%We write $X\sim P$ to denote that the random variable $X$ is drawn from the distribution $P$. 
Further, given reals $a,b$ with $a\leq b$, we define $\Pi_{[a,b]}(x):= \min\{\max\{x,a\},b\}$, for $x\in \mathbb{R}$. For a vector $\mathbf{x}\in \mathbb{R}^d$, we define its $\ell_p$-norm $\lVert \mathbf{x}\rVert_p$, for an integer $p\in \mathbb{N}$ to be $(\sum_{i=1}^d |x_i|^p)^{1/p}$, with $\lVert \mathbf{x}\rVert_\infty:= \max_{1\leq i\leq d} |x_i|$. For a given set $\mathcal{X}\subseteq \mathbb{R}^d$ and a vector $\mathbf{a}$, we define, with some abuse of notation, $\mathcal{X}(\mathbf{a})$ to be an $\ell_1$-projection $\arg \min_{\mathbf{b}\in \mathcal{X}} \lVert \mathbf{a}-\mathbf{b}\rVert_1$. Given an integer $d\geq 1$ and a real $M\geq 0$, we define $\Delta_M$ to be the $d$-simplex, i.e., $\Delta_M:= \{\mathbf{a}\in \mathbb{R}^d:\ \sum_{i\leq d} a_i\leq M,\ a_i\geq 0,\ \text{for all $i\in [d]$}\}$, and the set $\delta_M:= \{\mathbf{a}\in \mathbb{R}^d:\ \sum_{i\leq d} a_i= M,\ a_i\geq 0,\ \text{for all $i\in [d]$}\}$. Further, for $0\leq \alpha\leq \beta$, we define, with some abuse of terminology, the ``annulus" $\mathsf{A}_{\alpha,\beta}$ as $\mathsf{A}_{\alpha,\beta}:= (\Delta_\beta \setminus \Delta_\alpha) \cup \delta_\alpha$. 

We use the notation $\text{Lap}(b)$ to refer to the zero-mean Laplace distribution with standard deviation $\sqrt{2}b$, the notation $\text{Unif}((0,a])$ to denote the uniform distribution on the interval $(0,a]$, and the notation $\mathcal{N}(\mu, v)$ to denote the Gaussian distribution with mean $\mu$ and variance $v$.

%; its probability distribution function (p.d.f.) obeys
%$$
%f(x) = \frac{1}{2b}e^{-|x|/b}, \ x\in \mathbb{R}.
%$$

%Given a random variable $X$, we denote its variance by $\text{Var}(X)$. The notation ``$\log$'' denotes the natural logarithm.

%and $\mathcal{N}(\mu,\sigma^2)$ to denote the Gaussian distribution with mean $\mu$ and variance $\sigma^2$. For a random variable $X\sim \mathcal{N}(0,1)$, we denote its complementary cumulative distribution function (c.c.d.f.) by $Q$, i.e., for $x\in \mathbb{R}$,
%$
%Q(x):=\Pr[X\geq x] = \int_{x}^{\infty} \frac{1}{\sqrt{2\Pi}} e^{-z^2/2}\d z.
%$
%\subsection{Problem Setup}
%The ITMS dataset that we consider stores records of the data provided by IoT sensors deployed in an Indian city, containing, among other information, the license plate of the bus, the location at which the data was recorded, a timestamp, and the actual data value itself, which is the instantaneous speed of the bus. We process the dataset to extract the data records corresponding to that {h}exagonal ``grid'' in the city {a}nd that {t}imeslot 
%%(we abbreviate \underline{H}exagon \underline{A}nd \underline{T}imeslot  as HAT) 
%which has the highest traffic. 
%We remark that the algorithms discussed in this paper are applicable to general spatio-temporal IoT datasets too. 

\subsection{Problem Formulation}
\label{sec:problem}
Let $L$ be the number of users present in the dataset. For every user $\ell\in [L]$, let the number of records contributed by the user be denoted by $m_\ell$, and set $m^\star:= \max_{\ell\in [L]} m_\ell$ and $m_\star:= \min_{\ell\in [L]} m_\ell$. We assume that $L$ and the collection $\{m_\ell: \ell \in [L]\}$ are known. Now, for a given user $\ell\in [L]$, let $\left\{\mathbf{x}_j^{(\ell)}: j\in [m_\ell]\right\}$ denote the collection of (potentially arbitrary) bounded samples contributed by the user, where each $\mathbf{x}_j^{(\ell)}\in \mathbb{R}^d$, for some dimension $d\geq 1$. We assume, as is common for most applications \cite{dpsgd,dp_spcom}, that {$ \mathbf{x}_j^{(\ell)} \in \Delta_U$}, for all $\ell\in [L]$, $j\in [m_\ell]$, where $U> 0$ is known. Call the dataset as $\mathcal{D} = \left\{\left(\ell,\mathbf{x}_j^{(\ell)}\right): \ell \in [L], j\in [m_\ell]\right\}$.  
%The datasets of interest to us contain samples drawn according to some unknown distribution $P$ that is potentially non-i.i.d. (independent and identically distributed) across samples and users. 

We are interested in releasing the sample mean
\[
	\label{eq:f}
f = f(\mathcal{D}):= \frac{1}{\sum_{\ell'=1}^L m_{\ell'}}\cdot \sum_{\ell=1}^L \sum_{j=1}^{m_\ell} \mathbf{x}_j^{(\ell)}.
\]

%Call the dataset consisting of the speed records of users as $\mathcal{D} = \left\{\left(u_\ell,S_j^{(\ell)}\right): \ell \in [L], j\in [m_\ell]\right\}$, where the collection $\{u_\ell: \ell\in [L]\}$ denotes the set of users. 

%The function that we are interested in computing is the sample average
%$
%f(\mathcal{D}):= \frac{1}{\sum_{\ell=1}^L m_\ell}\cdot \sum_{\ell=1}^L \sum_{j=1}^{m_\ell} S_j^{(\ell)}.
%$

\subsection{Differential Privacy}
\label{sec:dp}
Consider datasets $\mathcal{D}_1 = \left\{\left({\ell},\mathbf{x}_j^{(\ell)}\right): \ell \in [L], j\in [m_\ell]\right\}$ and $\mathcal{D}_2 = \left\{\left(\ell,\overline{\mathbf{x}}_j^{(\ell)}\right): \ell \in [L], j\in [m_\ell]\right\}$ consisting of the same collection of users $[L]$, with each user contributing the same number, $m_\ell$, of  data values. Let $\mathsf{D}$ denote a universal set of such datasets. We say that $\mathcal{D}_1$ and $\mathcal{D}_2$ are ``user-level neighbours'' if there exists $\ell_0\in [L]$ such that $\left(\mathbf{x}^{(\ell_0)}_{1},\ldots, \mathbf{x}^{(\ell_0)}_{m_{\ell_0}}\right)\neq \left(\overline{\mathbf{x}}^{(\ell_0)}_{1},\ldots, \overline{\mathbf{x}}^{(\ell_0)}_{m_{\ell_0}}\right)$, with $\left(\mathbf{x}^{(\ell)}_{1},\ldots, \mathbf{x}^{(\ell)}_{m_{\ell}}\right)$ equal to $ \left(\overline{\mathbf{x}}^{(\ell)}_{1},\ldots, \overline{\mathbf{x}}^{(\ell)}_{m_{\ell}}\right)$, for all $\ell\neq \ell_0$. 
%In this work, we concentrate on mechanisms that map a given dataset to a single real value.

\begin{definition}
	For a fixed $\varepsilon>0$, a mechanism $M: \mathsf{D}\to \mathbb{R}^d$ is  user-level $\varepsilon$-DP if for every pair of user-level neighbours $\mathcal{D}_1, \mathcal{D}_2$ and for every measurable subset $Y \subseteq \mathbb{R}^d$, we have that
	$
	\Pr[M(\mathcal{D}_1) \in Y] \leq e^\varepsilon \Pr[M(\mathcal{D}_2) \in Y].
	$
\end{definition}
%Next, we define the user-level sensitivity of a function of interest.
\begin{definition}
	Given a function $g: \mathsf{D}\to \mathbb{R}^d$, we define its user-level sensitivity $\Delta_g$ as
	$
	\Delta_g:= \max_{\mathcal{D}_1,\mathcal{D}_2\ \text{u-l nbrs.}} \lVert g(\mathcal{D}_1) - g(\mathcal{D}_2)\rVert_1,
	$
	where the maximization is over datasets that are user-level neighbours.
\end{definition}
%For example, the user-level sensitivity of $f$ is
%$\Delta_f = \frac{Um^\star}{\sum_\ell m_\ell}$. 
We use the terms ``sensitivity'' and ``user-level sensitivity'' interchangeably. The next result is well-known \cite[Prop. 1]{dwork06}.

\begin{theorem}
	\label{thm:dp}
	For any $g: \mathsf{D}\to \mathbb{R}^d$, the mechanism $M^{\text{Lap}}_g: \mathsf{D}\to \mathbb{R}$ defined by
	$
	M^{\text{Lap}}_g(\mathcal{D}_1) = g(\mathcal{D}_1)+\mathbf{Z},
	$
	where $\mathbf{Z} = (Z_1,\ldots,Z_d)$ is such that $Z_i\stackrel{\text{i.i.d.}}{\sim} \text{Lap}(\Delta_g/\varepsilon)$, $i\in [d]$, and is independent of $\mathcal{D}$, is user-level $\varepsilon$-DP.
\end{theorem}

\subsection{The Worst-Case Error Metric}

All through, in this paper, we shall work with user-level $\varepsilon$-DP mechanisms that add a suitable amount of Laplace noise that is tailored to the sensitivity of the function used as an estimator of the sample mean $f$. %Now, given any mechanism $M$ for the $\varepsilon$-DP release of $f$, we shall first recall the formal definition of its worst-case error  \cite{tit-preprint}.
Consider a mechanism $M$ for the user-level $\varepsilon$-DP release of the statistic $f$. The canonical structure of $M$ (see \cite[Footnote 1]{staircase}, \cite{opt}) is: $M(\mathcal{D}) = \overline{f}(\mathcal{D})+{\mathbf{Z}},$
%\begin{equation}
%	\label{eq:Mbar}
%	M(\mathcal{D}) = \overline{f}(\mathcal{D})+{\mathbf{Z}},
%\end{equation}
for some estimator $\overline{f}$ of $f$, with user-level sensitivity $\Delta_{\overline{f}}$, with $\mathbf{Z} = (Z_1,\ldots,Z_d)$, with $Z_i\stackrel{\text{i.i.d.}}{\sim} \text{Lap}\left(\Delta_{\overline{f}}/\varepsilon\right)$, $i\in [d]$.  
%The assumption that our mechanisms are \emph{additive-noise} or \emph{noise-adding} mechanisms is without loss of generality, since it is known that every privacy-preserving mechanism can be thought of as a noise-adding mechanism (see \cite[Footnote 1]{staircase} and \cite{opt}). Moreover, under some regularity conditions, for small $\varepsilon$ (or equivalently, high privacy requirements), it is known that Laplace distributed noise is asymptotically optimal in terms of the magnitude of error in estimation \cite{staircase,opt}.

%Note that we work with the class of mechanisms that add Laplace noise tailored to the sensitivities of each grid, individually, since an explicit computation of the user-level sensitivity of the vector $f$ in \eqref{eq:f} (across all grids) is quite hard, thereby implying the necessity of loose bounds on the amount of noise added, when this notion of user-level sensitivity is used.
%Recall, from the discussion preceding Theorem \ref{thm:comp}, that the assumption that the mechanism ${}^gM_\theta$ is Laplace noise-adding is ca.

%Now, consider the mechanism $M_\theta$ that consists of the composition of the mechanisms ${}^gM_\theta$, over $g\in \mathcal{G}$, i.e., $M_\theta = \left({}^gM_\theta:\ g\in \mathcal{G}\right)$. In many settings of interest, a natural error metric for such a composition of mechanisms acting on different grids is the \emph{largest}  \emph{worst-case} estimation error among all the grids. 

For the mechanism $M$, we define its \emph{worst-case} estimation error as
\begin{equation}
	\label{eq:eg}
	E_M = E_{\overline{f}}^{(d)}:= \max_{\mathcal{D}\in \mathsf{D}} \big \lVert f(\mathcal{D}) - \overline{f}(\mathcal{D}) \big\rVert_1 + \E[\lVert{\mathbf{Z}}\rVert_1].
\end{equation}
Clearly, the expression above is an upper bound on the $\ell_1$-error $E^{(1)}_M:= \max_{\mathcal{D}\in \mathsf{D}} \E[|f(\mathcal{D})-M(\mathcal{D})|]$, via the triangle inequality. When the dimension $d$ is clear from the context, we simply denote $E_{\overline{f}}^{(d)}$ as $E_{\overline{f}}$. The distribution-independent expression in \eqref{eq:eg} conveniently captures the errors due to bias (the first term) and due to noise addition for privacy (the second term); a similar such error measure that separates the bias and noise errors was employed in \cite{amin}. 
%and its \emph{worst-case} estimation error over all datasets $\mathcal{D}$ as
%\[
%{}^gE := \max_{\mathcal{D}\in \mathsf{D}} {}^gE(\mathcal{D}).
%\]
%Finally, we define the error metric $E$ of the mechanism $M_\theta$ to be the \emph{largest} worst-case estimation error among all the grids, i.e., $$E := \max_{g\in \mathcal{G}} {}^gE.$$
\section{Worst-Case Errors of Bounding/Clipping Strategies}
\label{sec:bound}
%In this section, we present an explicit characterization of the worst-case errors for a broad class of strategies that work by ``bounding" or ``clipping" the contributions of users, as mentioned in the Introduction. In other words, the estimator $\overline{f}$ of $f$ in \eqref{eq:Mbar} is a suitably clipped or bounded version of the sample mean $f$. We then identify that clipping strategy that leads to the smallest worst-case error.
\subsection{On Clipping Strategies}
We work with estimators $\overline{f} = \overline{f}_{\left\{a_j^{(\ell)},b_j^{(\ell)}\right\}}$ of $f$ obtained by bounding user contributions as follows: for each $\ell\in [L]$ and $j\in [m_\ell]$, we let $\overline{\mathbf{x}}_j^{(\ell)}:= {\mathsf{A}_{a_j^{(\ell)},b_j^{(\ell)}}}\left(\mathbf{x}_j^{(\ell)}\right)$, for reals $0\leq a_j^{(\ell)}\leq b_j^{(\ell)}\leq U$. In words, $\overline{\mathbf{x}}_j^{(\ell)}$ is a projection of $\mathbf{x}_j^{(\ell)}$ onto the set $\mathsf{A}_{a_j^{(\ell)},b_j^{(\ell)}}$ (see Figure \ref{fig:proj}), which, intuitively, reduces the range of values that $\mathbf{x}_j^{(\ell)}$ can take, and hence its sensitivity too. We then set 
\[
\overline{f} = \overline{f}(\mathcal{D}):= \frac{1}{\sum_{\ell'=1}^L m_{\ell'}}\cdot \sum_{\ell=1}^L \sum_{j=1}^{m_\ell} \mathbf{\overline{x}}_j^{(\ell)}.
\]
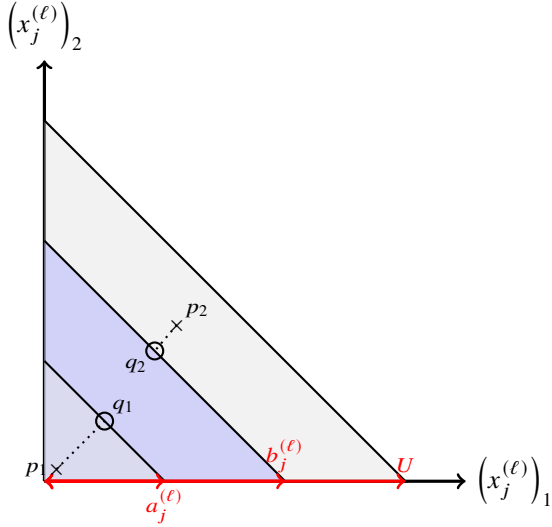
\begin{figure}
	\centering
	\begin{tikzpicture} [scale=1]% Adjust scale as needed
		
		\newcommand{\outerradius}{3}
		\newcommand{\middleradius}{2}
		\newcommand{\innerradius}{1}
		
		% Draw the axes with new labels (only positive quadrant)
		\draw[->, very thick] (0,0) -- (\outerradius+0.5,0) node[right] {$\left(x_j^{(\ell)}\right)_1$};
		\draw[->, very thick] (0,0) -- (0,\outerradius+0.5) node[above] {$\left(x_j^{(\ell)}\right)_2$};
		
		% Draw the balls with transparency (only positive quadrant)
		\draw[thick, fill=gray!20, fill opacity=0.5] (0,\outerradius) -- (\outerradius,0) -- (0,0);
		
		% Draw the blue region *first* (only positive quadrant)
		\fill[blue!30, fill opacity=0.5] (0,\middleradius) -- (\middleradius,0) -- (0,0);
		
		% Draw the inner ball - make it transparent gray (only positive quadrant)
		\draw[thick, fill=gray!20, fill opacity=0.5] (0,\innerradius) -- (\innerradius,0) -- (0,0); % Inner ball
		
		% Draw the middle ball outline in BLACK (only positive quadrant)
		\draw[thick,black] (0,\middleradius) -- (\middleradius,0) -- (0,0);
		
		\footnotesize
		% Add arrows and labels (red, bold tips)
		\draw[<->, very thick, red, line cap=round] (0,0) -- (\innerradius,0) node [below,font=\footnotesize] {$a_j^{(\ell)}$};
		\draw[<->, very thick, red, line cap=round] (0,0) -- (\middleradius,0) node [above,font=\footnotesize] {$b_j^{(\ell)}$};
		
		\draw[<->,very thick, red, line cap=triangle] (0,0) -- (\outerradius,0) node [above,font=\footnotesize] {$U$};
		
		% Mark the point (0.1, 0.1) and line segment
		\draw[thick, black] (0.1,0.1) node {\textbf{$\times$}} node[left] {$p_1$};
		\draw[dotted, thick] (0.1,0.1) -- (0.5,0.5);
		
		% Calculate the endpoint of the second line segment
		\pgfmathsetmacro{\endx}{1.1/1.2}
		\pgfmathsetmacro{\endy}{1.3/1.2}
		
		% Mark the point (1.1, 1.3)/1.2 with a circle
		\draw[thick, black] (\endx,\endy) circle (2pt) node[below left]{$q_2$};
		
		% Mark the point (1.1, 1.3)
		\draw[thick, black] (1.1,1.3) node {\textbf{$\times$}} node[above right]{$p_2$};
		
		% Draw the dotted line segment from (1.1, 1.3)
		\draw[dotted, thick] (1.1,1.3) -- (\endx,\endy);
		
		% Mark the point (0.5, 0.5) with a circle
		\draw[thick, black] (0.5, 0.5) circle (2pt) node[above right]{$q_1$};
		
	\end{tikzpicture}
	\caption{The annulus $\mathsf{A}_{a_j^{(\ell)},b_j^{(\ell)}}$, for $d=2$, shown in blue. Here, the points $q_1$, $q_2$ equal $\mathsf{A}_{a_j^{(\ell)},b_j^{(\ell)}}(p_1)$ and $\mathsf{A}_{a_j^{(\ell)},b_j^{(\ell)}}(p_2)$, respectively.}
	\label{fig:proj}
\end{figure}
Note that the class $\mathsf{B}$ of estimators $\overline{f}$ as above  captures those estimators obtained by dropping selected samples $\mathbf{x}_j^{(\ell)}$ (by setting $a_j^{(\ell)}= b_j^{(\ell)} = 0$ for those samples) and those obtained by projecting samples onto an $\ell_1$-bounded subset of $\Delta_U$, in addition to strategies that perform a combination of dropping and projection. This class of estimators hence includes several common estimators of the sample mean used in works such as \cite{userlevel,dp_preprint,tit-preprint}.

%In the next subsection, we discuss the worst-case errors (as in \eqref{eq:eg}) of mechanisms that employ some $\overline{f}\in \mathsf{B}$ as the estimator of the sample mean.

\subsection{On Worst-Case Errors of Clipping Strategies}
\label{sec:worst-case}
Consider the quantity $E_{\overline{f}}^{(d)}$, for $\overline{f}\in \mathsf{B}$. The following proposition then holds:
\begin{proposition}
	\label{lem:worst-case}
	We have that
	\begin{align*}
		E_{\overline{f}}^{(1)} \notag&= \frac{1}{\sum_{\ell'\leq L} m_{\ell'}}\cdot \left(\sum_{\ell\leq L}\sum_{j\leq m_\ell} \max\left\{a_j^{(\ell)},U-b_j^{(\ell)}\right\}\right)\ + \notag\\
		&\ \ \  \ \ \ \ \ \ \  \ \ \ \ \ \ \ \  \ \ \ \  \ \ \ \frac{d\cdot\max_{\ell\leq L} \sum_{j\leq m_\ell} \left(b_j^{(\ell)}-a_j^{(\ell)}\right)}{\varepsilon\cdot \sum_{\ell'\leq L} m_{\ell'}}
	\end{align*}
	and for any $d\geq 2$,
	\begin{align*}
		E_{\overline{f}}^{(d)} \notag&= \frac{1}{\sum_{\ell'\leq L} m_{\ell'}}\cdot \left(\sum_{\ell\leq L}\sum_{j\leq m_\ell} \max\left\{a_j^{(\ell)},U-b_j^{(\ell)}\right\}\right)\ + \notag\\
		&\ \ \  \ \ \ \ \ \ \  \ \ \ \ \ \ \ \  \ \ \ \  \ \ \ \ \ \ \ \ \ \ \frac{2d\cdot\max_{\ell\leq L} \sum_{j\leq m_\ell} b_j^{(\ell)}}{\varepsilon\cdot \sum_{\ell'\leq L} m_{\ell'}}.
	\end{align*}
\end{proposition}
The proof of Proposition \ref{lem:worst-case} proceeds with help from a few lemmas. For $E_{\overline{f}}$ as in \eqref{eq:eg}, we define $\beta^{(d)}(\overline{f}) = \beta(\overline{f})$ to be the bias, i.e., $\beta(\overline{f}):= \max_{\mathcal{D}\in \mathsf{D}} \big \lVert f(\mathcal{D}) - \overline{f}(\mathcal{D}) \big\rVert$ and $\eta^{(d)}(\overline{f})=\eta(\overline{f})$ to be the error due to noise addition, i.e., $\eta(\overline{f}):= \E[\lVert\mathbf{\mathbf{Z}}\rVert_1]$. First, we aim to characterize $\beta(\overline{f})$. To this end, we first state a necessary condition for a vector $\mathbf{y}$ to be an $\ell_1$-projection of a vector $\mathbf{a}\in \Delta_U$ onto $\Delta_\alpha$, for $ \alpha\leq U$. 
\begin{lemma}
	\label{lem:proj1}
	Given $\mathbf{a}\in \Delta_U$ and $\alpha\leq U$, we have $\Delta_\alpha(\mathbf{a}) = \mathbf{a}$, if $\mathbf{a}\in \Delta_{\alpha}$. Else, any $\ell_1$-projection $\mathbf{y} = \Delta_\alpha(\mathbf{a})$ must satisfy $\lVert \mathbf{y}\rVert_1 = \alpha$, with $y_i\leq a_i$, for all $i\in [d]$.
\end{lemma}
\begin{proof}
	The first statement in the lemma is clear. Now, for $\mathbf{a}\notin \Delta_\alpha$, suppose that $\mathbf{y} = \Delta_\alpha(\mathbf{a})$ is such that $\lVert \mathbf{y}\rVert_1 < \alpha$. It can then be seen that by setting $\mathbf{y}':= \lambda\mathbf{y}+(1-\lambda)\mathbf{a}$, for some $\lambda\in (0,1)$ such that $\lVert \mathbf{y}'\rVert_1 = \alpha$, we will obtain that $\lVert \mathbf{y}'-\mathbf{a}\rVert_1 < \lVert \mathbf{y}-\mathbf{a}\rVert_1$, which is a contradiction. Likewise, suppose that $y_i>a_i$, for some $i\in [d]$. Now, consider any coordinate $j\in [d]$ such that $y_j\leq a_j$ (such a coordinate must exist, since $\lVert \mathbf{y}\rVert_1 = \alpha$); by letting $m:= \min\{|y_i-a_i|,|y_j-a_j|\}$ and setting $y_i \gets y_i-m$ and $y_j \gets y_j+m$, we see that we strictly decrease $\lVert \mathbf{y}-\mathbf{a}\rVert_1$, which is a contradiction.
\end{proof}
In the appendix, we explicitly characterize the vectors $\mathbf{y}\in \Delta_\alpha$ that are the $\ell_1$-projections of $\mathbf{a}\in \Delta_U\setminus \Delta_\alpha$, for $\alpha\leq U$; indeed, we have that the condition stated in Lemma \ref{lem:proj1} is both necessary and sufficient. The next lemma exactly characterizes the bias $\beta^{(d)}(\overline{f})$, for any $d\geq 1$.
\begin{lemma}
	\label{lem:proj2}
	We have that for any $d\geq 1$,
	\[\beta^{(d)}(\overline{f}) = \frac{1}{\sum_{\ell'\leq L} m_{\ell'}}\cdot \left(\sum_{\ell\leq L}\sum_{j\leq m_\ell} \max\left\{a_j^{(\ell)},U-b_j^{(\ell)}\right\}\right).\]
\end{lemma}
\begin{proof}
	For a sample $\mathbf{x}_j^{(\ell)}\in \mathsf{A}_{a_j^{(\ell)},b_j^{(\ell)}}$, it is clear that $\mathsf{A}_{a_j^{(\ell)},b_j^{(\ell)}}(\mathbf{x}_j^{(\ell)}) = \mathbf{x}_j^{(\ell)}$. Now, suppose that $\mathbf{x}_j^{(\ell)}\in \Delta_U\setminus \Delta_{b_j^{(\ell)}}$. Following Lemma \ref{lem:proj1}, if $\mathbf{y} = \mathsf{A}_{a_j^{(\ell)},b_j^{(\ell)}}(\mathbf{a})$, we must have $$\lVert \mathbf{y} - \mathbf{x}_j^{(\ell)}\rVert_1 = \sum_{i\leq d} ((x_j^{(\ell)})_i - y_i) = \lVert \mathbf{x}_j^{(\ell)}\rVert_1 - b_j^{(\ell)}.$$ Hence, the worst-case clipping error for sample $\mathbf{x}_j^{(\ell)}$ is $$\max\ \lVert \mathsf{A}_{a_j^{(\ell)},b_j^{(\ell)}}(\mathbf{x}_j^{(\ell)}) - \mathbf{x}_j^{(\ell)}\rVert_1 = U-b_j^{(\ell)}.$$ By symmetric arguments, one can show that if $\mathbf{x}_j^{(\ell)}\in \Delta_{a_j^{(\ell)}}\setminus \delta_{a_j^{(\ell)}}$, we must have $$\max\ \lVert \mathsf{A}_{a_j^{(\ell)},b_j^{(\ell)}}(\mathbf{x}_j^{(\ell)}) - \mathbf{x}_j^{(\ell)}\rVert_1 = a_j^{(\ell)}.$$ Thus, overall, we obtain that 
	\begin{equation}\max_{\mathbf{x}_j^{(\ell)}\in \Delta_U} \lVert \mathsf{A}_{a_j^{(\ell)},b_j^{(\ell)}}(\mathbf{x}_j^{(\ell)}) - \mathbf{x}_j^{(\ell)}\rVert_1 = \max\left\{a_j^{(\ell)},U-b_j^{(\ell)}\right\}. \label{eq:maxclip}\end{equation} 
	Now, for any dataset $\mathcal{D}$, recall that
	\begin{align}
		&\lVert f(\mathcal{D})-\overline{f}(\mathcal{D})\rVert_1 \notag\\
		%&=\frac{1}{\sum_{\ell'\leq L} m_{\ell'}}\left\lVert \sum_{\ell\leq L}\sum_{j\leq m_\ell} \left(\mathbf{x}_j^{(\ell)} - \overline{\mathbf{x}}_j^{(\ell)}\right)\right \rVert_1 \notag\\
		&=\frac{1}{\sum_{\ell'\leq L} m_{\ell'}}\left\lVert \sum_{\ell\leq L}\sum_{j\leq m_\ell} \left(\mathbf{x}_j^{(\ell)} - \mathsf{A}_{a_j^{(\ell)},b_j^{(\ell)}}(\mathbf{x}_j^{(\ell)}) \right)\right \rVert_1. \label{eq:cliphelp}
	\end{align}
	Putting together \eqref{eq:maxclip} and \eqref{eq:cliphelp} concludes the proof.
%	It is then easy to see that the bias
%	\begin{align}
%		&\max_{\mathcal{D}\in \mathsf{D}} |f(\mathcal{D})-\overline{f}(\mathcal{D})| \notag\\&= \frac{1}{\sum_{\ell'\leq L} m_{\ell'}}\cdot \left(\sum_{\ell\leq L}\sum_{j\leq m_\ell} \max\left\{a_j^{(\ell)},U-b_j^{(\ell)}\right\}\right). \label{eq:bias}
%	\end{align}
\end{proof}
The calculation of $\eta(\overline{f})$ is quite similar to the proof above, and is captured in Lemma \ref{lem:proj3} below.

\begin{lemma}
	\label{lem:proj3}
	We have that $$\eta^{(1)}(\overline{f}) = \frac{d\cdot \max_{\ell\leq L} \sum_{j\leq m_\ell} \left(b_j^{(\ell)}-a_j^{(\ell)}\right)}{\varepsilon\cdot \sum_{\ell'\leq L} m_{\ell'}}$$
	and for any $d\geq 2$,
	$$\eta^{(d)}(\overline{f}) = \frac{2d\cdot \max_{\ell\leq L} \sum_{j\leq m_\ell} b_j^{(\ell)}}{\varepsilon\cdot \sum_{\ell'\leq L} m_{\ell'}}.$$
\end{lemma}
\begin{proof}
Consider first the case when $d = 1$. Via arguments entirely analogous to that in the proof of Lemma \ref{lem:proj2}, the user-level sensitivity of $\overline{f}$ is
$$
	\Delta_{\overline{f}} = \frac{\max_{\ell\leq L} \sum_{j\leq m_\ell} \left(b_j^{(\ell)}-a_j^{(\ell)}\right)}{\sum_{\ell'\leq L} m_{\ell'}},
$$
since in the worst-case, all samples of a user $\ell$ are changed each from $(a_j^{(\ell)},0,\ldots,0)$ to $(b_j^{(\ell)},0,\ldots,0)$. Next, for any $d\geq 2$, we have via similar arguments that
$$
\Delta_{\overline{f}} = \frac{2\max_{\ell\leq L} \sum_{j\leq m_\ell} b_j^{(\ell)}}{\sum_{\ell'\leq L} m_{\ell'}},
$$
since in the worst-case, all samples of a user $\ell$ are changed each from $(b_j^{(\ell)},0,0,\ldots,0)$ to $(0,b_j^{(\ell)},0,\ldots,0)$.
Using $\E[\lVert \mathbf{Z}\rVert_1] = d\E[|Z_1|] = d\Delta_{\overline{f}} /\varepsilon$ gives us the required result.
\end{proof}
The proof of Proposition \ref{lem:worst-case} then follows directly by putting together Lemmas \ref{lem:proj2} and \ref{lem:proj3}.
%\begin{remark}
%	Via the proof of Lemma \ref{lem:worst-case}, we observe that for a given distribution $\{m_\ell:\ell\in [L]\}$ of user contributions, the dataset $\mathcal{D}$ that gives rise to the largest (or worst-case) error is that where $x_j^{(\ell)} = U$, for all $\ell\in [L]$, $j\in [m_\ell]$. We shall use this observation in the next section when we numerically evaluate the performance of the clipping strategy in \cite{amin} on this ``worst-case dataset".
%\end{remark}

Given the characterization of the worst-case error $E_{\overline{f}}$ as above, we proceed with identifying an estimator $f^\star \in \mathsf{B}$ (equivalently, a bounding strategy) that minimizes $E_{\overline{f}}$, over all $\overline{f}\in \mathsf{B}$. Let $T_\varepsilon$ denote the $\left \lceil \left(\frac{2d}{\varepsilon}\right)\right \rceil^{\text{th}}$-largest value in the collection $\{Um_1,Um_2,\ldots,Um_L\}$; if $\varepsilon<2d/L$, we set $T_\varepsilon = 0$. Our main result is encapsulated in the following theorem:
\begin{theorem}
	\label{thm:opt}
	We have that $f^\star = \overline{f}_{\left\{a_j^{(\ell)},b_j^{(\ell)}\right\}}$ minimizes $E_{\overline{f}}$, where
	$
	a_j^{(\ell)} = \max\left\{\frac{Um_\ell-T_\varepsilon}{2m_\ell},0\right\}\text{ and } b_j^{(\ell)} = \min\left\{\frac{Um_\ell+T_\varepsilon}{2m_\ell},U\right\},
	$
	if $d=1$, and $
	a_j^{(\ell)} = 0\ \text{ and } b_j^{(\ell)} =\min\left\{\frac{T_\varepsilon}{m_\ell},U\right\}
	$, for any $d\geq 2$, for all $\ell\in [L]$ and $j\in [m_\ell]$.
\end{theorem}
Some remarks are in order. First, note that the optimal bounding strategy $f^\star$ clips sample values based only on the number of contributions of each user $\ell\in [L]$. Furthermore, note that the interval of projection is determined by $T_\varepsilon$, which is quite similar in structure to the optimal clipping threshold $T$ in \cite{amin} for item-level DP, which is the (privately estimated) $\left \lceil \left(\frac{2}{\varepsilon}\right)\right \rceil^{\text{th}}$-largest \emph{sample value}.

The proof of Theorem \ref{thm:opt} proceeds with help from the following lemma.
\begin{lemma}
	\label{lem:helper}
	For $d = 1$, there exists an estimator $\overline{f}_{\left\{a_j^{(\ell)},b_j^{(\ell)}\right\}}$ minimizing $E_{\overline{f}}$, which obeys $a_j^{(\ell)}+b_j^{(\ell)} = U$, for all $\ell\in [L]$ and $j\in [m_\ell]$. Furthermore, for $d\geq 2$, an estimator $\overline{f}_{\left\{a_j^{(\ell)},b_j^{(\ell)}\right\}}$ minimizing $E_{\overline{f}}$ must set $a_j^{(\ell)} = 0$, for all $\ell\in [L]$ and $j\in [m_\ell]$.
\end{lemma}
\begin{proof}
	Consider any optimal estimator $\overline{f}_{\left\{a_j^{(\ell)},b_j^{(\ell)}\right\}}$, and suppose that $a_j^{(\tilde{\ell})}+b_j^{(\tilde{\ell})}>U$, for some $\tilde{\ell}\in [L]$ and $j\in [m_{\tilde{\ell}}]$. The proof when we suppose that $a_j^{(\tilde{\ell})}+b_j^{(\tilde{\ell})}<U$, for some $\tilde{\ell}\in [L]$ and $j\in [m_{\tilde{\ell}}]$, is similar, and is omitted. Let $\delta =  (a_j^{(\tilde{\ell})}+b_j^{(\tilde{\ell})})-U$. By setting $b_j^{(\tilde{\ell})}\gets b_j^{(\tilde{\ell})}-\delta$, we observe that $\beta^{(d)}(\overline{f})$ remains unchanged, while $\eta^{(d)}(\overline{f})$ either remains unchanged or strictly decreases by $\delta>0$. Moreover, for $d\geq 2$, note from Proposition \ref{lem:worst-case} that increasing any $a_j^{(\ell)}$ from $0$ to a positive value can only strictly increase $E_{\overline{f}}$.
\end{proof}
Hence, to obtain the explicit structure of an estimator $\overline{f}$ that minimizes $E_{\overline{f}}$, it suffices to focus estimators $\overline{f}_{\left\{a_j^{(\ell)},b_j^{(\ell)}\right\}}$ with $a_j^{(\ell)}+b_j^{(\ell)} = U$, for all $\ell\in [L]$, $j\in [m_\ell]$. Then, for $d=1$, there exists an estimator $f^\star$ minimizing $E^{(1)}_{\overline{f}}$ that satisfies
\begin{equation}
	\label{eq:optsimp1}
	E^{(1)}_{f^\star} = \frac{1}{\sum_{\ell'\leq L} m_{\ell'}}\cdot \left(\sum_{\ell\leq L} S^{(\ell)}+\frac{1}{\varepsilon}\cdot \max_{\ell\leq L} \left(Um_\ell - 2S^{(\ell)}\right)\right),
\end{equation}
where for $\ell\in [L]$, we let $S^{(\ell)}:= \sum_{j\leq m_\ell} a_j^{(\ell)}$, for any given estimator $\overline{f}$. Furthermore, for any $d\geq 2$, there exists an estimator $f^\star$ minimizing $E^{(d)}_{\overline{f}}$ that satisfies
\begin{equation}
	\label{eq:optsimpd}
	E^{(d)}_{f^\star} = \frac{1}{\sum_{\ell'\leq L} m_{\ell'}}\cdot \left(\sum_{\ell\leq L} \overline{S}^{(\ell)}+\frac{2d}{\varepsilon}\cdot \max_{\ell\leq L} \left(Um_\ell - \overline{S}^{(\ell)}\right)\right),
\end{equation}
where for $\ell\in [L]$, we let $\overline{S}^{(\ell)}:= \sum_{j\leq m_\ell} \left(U - b_j^{(\ell)}\right)$, for any given estimator $\overline{f}$. We now prove Theorem \ref{thm:opt}.

\begin{proof}[Proof of Thm. \ref{thm:opt}]
	We prove the theorem in detail for the case when $d=1$; the setting of $d\geq 2$ is quite similar and the proof is hence is omitted. We begin with the expression in \eqref{eq:optsimp1}; note that our task now is simply to identify the optimal parameters $\{S^{(\ell)}:\ \ell\in [L]\}$ of $f^\star$ in \eqref{eq:optsimp1}. Once these parameters are derived, we simply set $a_j^{(\ell)}:= S^{(\ell)}/m_\ell = U-b_j^{(\ell)}$ (via Lemma \ref{lem:helper}), for each $\ell\in [L]$ and $j\in [m_\ell]$. In the expression in \eqref{eq:optsimp1}, let us set $\tau:= \max_{\ell\leq L} \left(Um_\ell - 2S^{(\ell)}\right)$. We hence need to solve the following constrained optimization problem:
	\begin{align}
		&{\text{minimize}}\quad h(\{S^{(\ell)}\}):= \left(\sum_{\ell\leq L} S^{(\ell)}+\frac{\tau}{\varepsilon}\right)\notag\\
		&\text{subj. to:}\ \  Um_\ell-2S^{(\ell)}\leq \tau,\ S^{(\ell)}\geq 0,\ \forall\ \ell\in [L],\ \tau\geq 0.\label{eq:opt}
	%	&\ \ \ \  \ \ \ \  \ \ \ \ \ \ S\geq 0. \label{eq:opt}
	\end{align}
The optimization problem in \eqref{eq:opt} is a linear programming problem. By standard arguments via the necessity of the KKT conditions \cite[Sec. 5.5.3]{boyd}, there must exist reals $\lambda_\tau\geq 0$ and $\lambda_\ell,\ \mu_\ell\geq 0$, for each $\ell\in [L]$ (or Lagrange multipliers), such that the function 
%$$\mathcal{L}(\{S^{(\ell)}\},\ \lambda_\tau,\{\lambda_\ell,\mu_\ell\}):= \sum_{\ell\leq L} S^{(\ell)}+\frac{d\tau}{\varepsilon}-\lambda_\tau \tau\ + \\
%\sum_\ell \lambda_\ell\cdot \left(Um_\ell-2S^{(\ell)}-\tau\right)-\sum_\ell \mu_\ell S^{(\ell)}$$
\begin{align*}
\mathcal{L}(\{S^{(\ell)}\},\ &\lambda_\tau,\{\lambda_\ell,\mu_\ell\}):= \sum_{\ell\leq L} S^{(\ell)}+\frac{S}{\varepsilon}-\lambda_\tau \tau\ + \\
&\sum_\ell \lambda_\ell\cdot \left(Um_\ell-2S^{(\ell)}-S\right)-\sum_\ell \mu_\ell S^{(\ell)}
\end{align*}
obeys the following properties.
\begin{itemize}
	\item \textbf{Stationarity}: We have that $\frac{\partial \mathcal{L}}{\partial \tau} = 0$, or $\lambda_\tau+\sum_\ell \lambda_\ell = \frac{d}{\varepsilon},$
%	\begin{equation}
%		\label{eq:stat1}
%		\lambda_S+\sum_\ell \lambda_\ell = \frac{1}{\varepsilon},
%	\end{equation}
and that $\frac{\partial \mathcal{L}}{\partial S^{(\ell)}} = 0$, for each $\ell\in [L]$, or $\lambda_\ell = \frac{1-\mu_\ell}{2}.$
%\begin{equation}
%	\label{eq:stat2}
%	\lambda_\ell = \frac{1-\mu_\ell}{2}.
%\end{equation}
\item \textbf{Complementary slackness}: We have that $$\lambda_\tau \tau= 0,\ \lambda_\ell\cdot \left(Um_\ell-2S^{(\ell)}-\tau\right) = 0,\ \text{and}\ \mu_\ell S^{(\ell)} = 0,$$
%\begin{align}
%	\label{eq:comp}
%	\lambda_SS = 0,\ \lambda_\ell\cdot \left(Um_\ell-2S^{(\ell)}-S\right) = 0,\ \text{and}\ \mu_\ell S^{(\ell)} = 0,
%\end{align}
for all $\ell\in [L]$.
\end{itemize}
We claim that the assignment $\tau^\star = T_\varepsilon$ and $S^{(\ell),\star} = \max\left\{\frac{Um_\ell-T_\varepsilon}{2},0\right\}$ satisfies the conditions above, for an appropriate choice of $\lambda_\tau,\{\lambda_\ell,\mu_\ell\}$ values. Indeed, for $\ell\in [L]$, if $S^{(\ell),\star} = 0$, we set $\lambda_\ell = 0$ and $\mu_\ell = 1$; else, we set $\lambda_\ell = \frac12$ and $\mu_\ell = 0$. 

%For the setting where $d\geq 2$, the optimization problem \eqref{eq:opt} is modified to the following convex optimization program:
%	\begin{align}
%	&{\text{minimize}}\quad h^{(d )}(\{S^{(\ell)}\}):= \left(\sum_{\ell\leq L} \overline{S}^{(\ell)}+\frac{2d\tau^{(d)}}{\varepsilon}\right)\notag\\
%	&\text{subj. to:}\ \  Um_\ell-2\overline{S}^{(\ell)}\leq \tau^{(d)},\ \overline{S}^{(\ell)}\geq 0,\ \forall\ \ell\in [L], \notag\\
%	&\ \ \ \  \ \ \ \  \ \ \ \ \tau^{(d)}\geq 0. \label{eq:optnew}
%	%	&\ \ \ \  \ \ \ \  \ \ \ \ \ \ S\geq 0. \label{eq:opt}
%\end{align}
%In \eqref{eq:optnew}, we set $\tau^{(d)}:= \max_{\ell\leq L} \left(Um_\ell - \overline{S}^{(\ell)}\right)$. This new optimization problem can be solved similarly using the KKT optimality conditions, to yield the optimal values of $\{a_j^{(\ell)},b_j^{(\ell)}\}$ as in the statement of the theorem.
%An optimal choice of parameters $\{a_j^{(\ell)},b_j^{(\ell)}\}$ of an estimator $f^\star$ thus obeys $a_j^{(\ell)}:= \alpha^{(\ell),\star}/m_\ell = U-b_j^{(\ell)}$, for each $\ell\in [L]$ and $j\in [m_\ell]$.
\end{proof}
Observe that from Theorem \ref{thm:opt}, the optimal worst-case error when $d=1$ is \begin{equation}E^\text{OPT, $(1)$}(\varepsilon)= \frac{1}{\sum_{\ell'\leq L} m_{\ell'}}\cdot\left(\sum_{\ell\leq L} \max\left\{\frac{Um_\ell-T_\varepsilon}{2},0\right\}+\frac{T_\varepsilon}{\varepsilon}\right). \label{eq:eopt1}\end{equation} %
and the the optimal worst-case error for any $d\geq2$ is such that 
\begin{align}& E^\text{OPT, $(d)$}(\varepsilon) \notag\\
	&= \frac{1}{\sum_{\ell'\leq L} m_{\ell'}}\cdot \left( \sum_{\ell\leq L} \max\left\{{Um_\ell-T_\varepsilon},0\right\}+\frac{2dT_\varepsilon}{\varepsilon}\right). \label{eq:eoptd}\end{align} %
We end this section with a couple of remarks. First, observe from \eqref{eq:eopt1} and \eqref{eq:eoptd} that, in the limit as $\varepsilon\to 0$, the optimal worst-case error equals the worst-case bias error, which in turn equals $U/2$, for $d=1$, and $U$, when $d\geq 2$.

Next, consider the special case of Theorem \ref{thm:opt} when $d=1$, with the interpretation that each sample $x_j^{(\ell)}\in [0,U]$. Here, the annulus $\mathsf{A}_{a_j^{(\ell)},b_j^{(\ell)}}$ is simply the interval $[a_j^{(\ell)}, b_j^{(\ell)}]\subseteq [0,U]$. This implies that for the case when the samples $\mathbf{x}_j^{(\ell)}$ are allowed to take values in the cube $[0,U]^d$, as against in the $\ell_1$-ball $\lVert \mathbf{x}_j^{(\ell)}\rVert_1\leq U$, one can perform the bounding procedure, using the \emph{same} interval $[a_j^{(\ell)}, b_j^{(\ell)}]$ identified via Theorem \ref{thm:opt} for the $d=1$ setting, for each dimension, independently.

\section{Numerical Experiments}
\label{sec:experiments}
In this section, we compare, via numerical experiments, the performance of our clipping strategy for the case when $d=1$, with the widely used strategy in \cite[Sec. 3]{amin}, which we call the ``AKMV" mechanism\footnote{The subscript ``AKMV'' stands for the initials of the last names of the authors of \cite{amin}.}. We straightforwardly adapt the mechanism in \cite{amin} from the item-level setting to the user-level setting, by considering the contribution of a user to be the sum $\sum_{j\leq m_\ell} x_j^{(\ell)}$ of the samples it contributes in the user-level setting.

Importantly, the AKMV mechanism uses a privacy budget of $\varepsilon/2$ to first \emph{privately} estimate the $\left \lceil \left(\frac{2}{\varepsilon}\right)\right \rceil^{\text{th}}$-largest sample, and then uses this value as the clipping threshold. Via entirely heuristic analysis, the authors of \cite{amin} argue the ``optimality" of AKMV mechanism. However, importantly, their analysis ignores the error in the private estimation of the the $\left \lceil \left(\frac{2}{\varepsilon}\right)\right \rceil^{\text{th}}$-largest sample, rendering such a claim of optimality incorrect. As we shall see, the clipping strategy in \cite[Sec. 3]{amin} performs quite poorly in comparison to our worst-case-error-optimal strategy, in the average case too, when the data samples are drawn i.i.d. from natural distributions.

\subsection{Experimental Setup}
\label{sec:exp-setup}
A natural application of user-level DP mechanisms is to spatio-temporal datasets; similar to previous work \cite{dp_spcom,dp_preprint}, we let $U = 65$, in line with the largest speed of buses in km/hr, in Indian cities. We consider two collections of numbers of user contributions $\{m_\ell:\ \ell\in [L]\}$. 
%We call our first collection loosely as the ``geometric collection", owing to the nature of the distribution of the $m_\ell$ values, and our second collection as the ``extreme-valued collection".
\begin{enumerate}
	\item Geometric collection: Here, we fix an integer $M$ and consider $L = 2^{M+1}-1$ users; for each $i\in \{0,1,\ldots,M\}$, we create $2^i$ users each contributing $2^{M-i}$ samples. 
	%It is clear that there are $2^{M+1}-1$ users overall; further, the total number of samples is $\sum_\ell m_\ell = (M+1)\cdot 2^M$ and $m^\star = 2^M$. 
	In our experiments, we set $M = 6$.
	\item Extreme-valued collection: Here, we consider $L$ users where $L-1$ users contribute one sample each and $1$ user contributes $m^\star>1$ samples. %The total number of samples is $\sum_\ell m_\ell = L-1+m^\star$. 
	In our experiments, we pick $L = 101$, with $m^\star = 10$.
\end{enumerate}
%For a fixed collection $\{m_\ell:\ \ell \in [L]\}$, we recall from the remark following  the proof of Lemma \ref{lem:worst-case}) that the worst-case dataset, i.e., the dataset $\mathcal{D}$ that attains the worst-case error in Lemma \ref{lem:worst-case} is the one where all samples are equal to $U$. 
We also work with the following synthetically generated datasets with i.i.d. samples:
\begin{enumerate}
	\item Uniform samples: Each sample $x_j^{(\ell)}\stackrel{\text{i.i.d}}{\sim} \text{Unif}((0,U])$, across $\ell\in [L]$, $j\in [m_\ell]$. 
	\item Projected Gaussian samples: Each sample $x_j^{(\ell)}$, $\ell\in [L]$, $j\in [m_\ell]$, is drawn i.i.d. by rejection sampling from the $\mathcal{N}(U/2,U/4)$ distribution so that the samples lie in $(0,U]$. 
\end{enumerate}
\subsection{Performance on Synthetic Datasets}
%Given the estimators $f, f^\star, \widehat{g}_T$ as in \eqref{eq:f}, \eqref{eq:fclip}, \eqref{eq:ghat}, respectively, we first put down their sensitivities. Indeed, via straightforward computations, we have that
%\begin{equation}
%	\label{eq:sens1}
%	\Delta_f = \frac{Um^\star}{\sum_\ell m_\ell},\ \Delta_{f^\star} = \frac{\max_\ell \sum_{j\leq m_\ell}(b_j^{(\ell)}-a_j^{(\ell)})}{\sum_\ell m_\ell},\ 
%\end{equation}
%{and}
%\begin{equation}
%	\label{eq:sens2}
%	\Delta_{\widehat{g}_T} = \frac{T}{\sum_\ell m_\ell}.
%\end{equation}
%In \eqref{eq:sens1}, the values $\{a_j^{(\ell)}\}, \{b_j^{(\ell)}\}$ are from Theorem \ref{thm:opt} and in \eqref{eq:sens2}, the threshold $T$ is computed (privately) via Step 2 in Section \ref{sec:amin}.
 
We compare the performances of the following three mechanisms on i.i.d. synthetic datasets:
\begin{enumerate}
	\item the vanilla Laplace mechanism that releases $M_{\text{Lap}}(\mathcal{D}) = f(\mathcal{D})+\mathbf{Z}_1$, where $\mathbf{Z}_1\sim \text{Lap}\left(\frac{\Delta_f}{\varepsilon}\right)$;
	\item the ``OPT-worst-case" mechanism that releases $M_{\text{OPT-wc}}(\mathcal{D}) = f^\star(\mathcal{D})+\mathbf{Z}_2$, where $\mathbf{Z}_2\sim \text{Lap}\left(\frac{\Delta_{f^\star}}{\varepsilon}\right)$; and
	\item the ``AKMV" mechanism \cite[Sec. 3]{amin} that releases $M_{\text{AKMV}}(\mathcal{D}) = \widehat{f}_T(\mathcal{D})+\mathbf{Z}_3$, where $\mathbf{Z}_3\sim \text{Lap}\left(\frac{2\Delta_{\widehat{g}_T}}{\varepsilon}\right)$, and $\widehat{f}_T(\mathcal{D})$ is obtained by clipping the sum of samples of each user to lie in $[0,T]$. Here, $T$ is the estimate of the $\left \lceil \left(\frac{2}{\varepsilon}\right)\right \rceil^{\text{th}}$-largest sample among $\{\sigma^{(\ell)}\}$, which is estimated with privacy budget $\varepsilon/2$.
\end{enumerate}
%\begin{itemize}
%	\item 
%	\item 
%	\item 
%\end{itemize}

The`` average-case" errors on i.i.d. datasets are defined as
$$\overline{E}_\text{Lap} = \E[|\mathbf{Z}_1|],\ \overline{E}_\text{OPT-wc} = \E[|M_{\text{OPT-wc}}(\mathcal{D}) - f(\mathcal{D})|],$$
and
$$ \overline{E}_\text{AKMV} = \E[|M_{\text{AKMV}}(\mathcal{D}) - f(\mathcal{D})|],$$
where the expectations are over the randomness in the data samples and in the DP mechanism employed. These errors 
are then estimated via Monte-Carlo averaging over $10^4$ iterations. As a pre-processing step, we replace each of the samples $\{x_j^{(\ell)}:\ j\in [m_\ell]\}$ of every user $\ell \in [L]$ by the sample average $\frac{1}{m_\ell}\cdot \sum_{j\leq m_\ell} x_j^{(\ell)}$, so as to allow for improved performance post clipping. 

Figure \ref{fig:avg-geo} and \ref{fig:avg-extreme-normal} (with the error axis on a log-scale), respectively, show plots of the (estimates of) the average-case errors for the three mechanisms above for the geometric collection of $\{m_\ell\}$ values and uniform samples, and for the extreme-valued collection with projected Gaussian samples. Interestingly, the average-case performance of the AKMV mechanism is quite similar to, and a little worse than, the average-case performance of the vanilla Laplace mechanism, which in turn is significantly worse than the average-case error incurred by the OPT-worst-case mechanism. A reason for the poor performance of the AKMV mechanism overall is due to the reduced privacy budget allocated to the private release of the clipped estimator. Moreover, in contrast to the OPT-worst-case mechanism, the AKMV mechanism requires the private estimation of the $\left \lceil \left(\frac{2}{\varepsilon}\right)\right \rceil^{\text{th}}$ quantile value, which incurs additional error. We expect that similar trends can be observed for i.i.d. datasets with samples drawn from other distributions of $\{m_\ell\}$ values and samples. We add that while the clipping error remains roughly the same for small changes in $\varepsilon$, there exist certain values of $\varepsilon$ that give rise to sharp discontinuities in the $T_\varepsilon$ values; such behaviour hence results in the (estimates of the) average-case errors not being monotonic in $\varepsilon$, in Figure \ref{fig:avg-geo}.

% Figure \ref{fig:avg-extreme-normal} shows plots of the (estimates of) the average-case errors for the extreme-valued collection of $\{m_\ell\}$ values with projected Gaussian samples. In our experiments, we pick $L = 101$, with $m^\star = 10$, and the errors are shown on a log-scale for easy viewing. These plots show a similar trend in the performances of the three mechanisms as that in Figure \ref{fig:avg-geo}. 
\begin{figure}[!t]
	\centering
	\includegraphics[width = \linewidth]{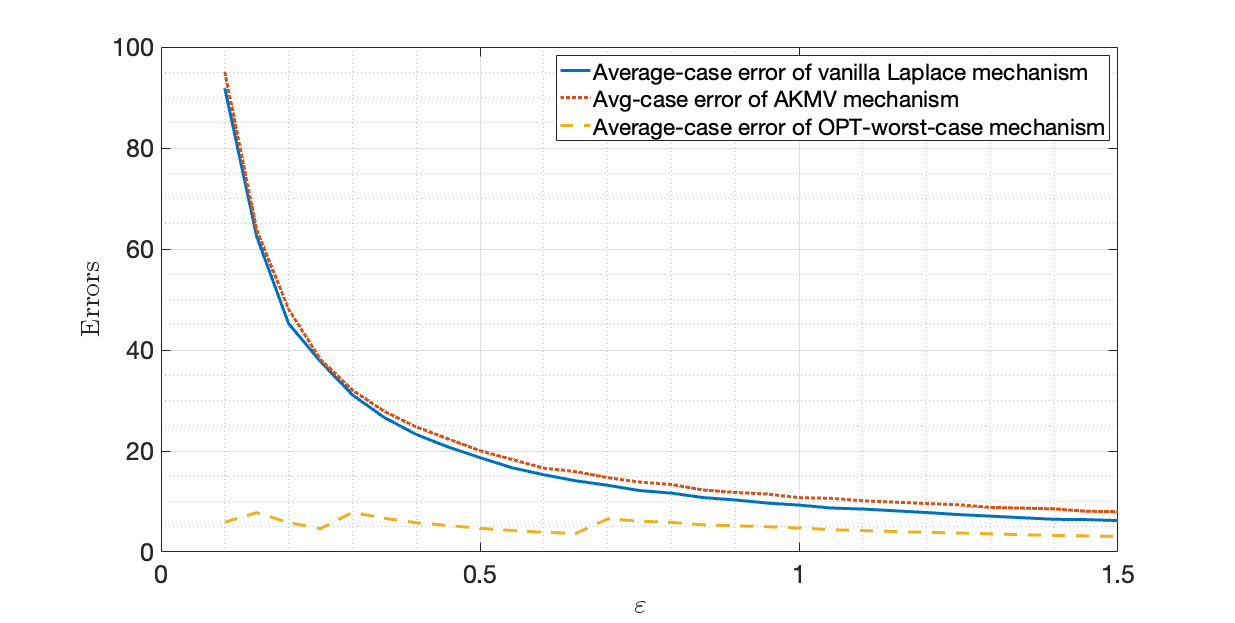}
	\caption{Average-case errors using a geometric collection of $\{m_\ell\}$ values and uniform samples}
	\label{fig:avg-geo}
\end{figure}
%
%\begin{figure}[!t]
%	\centering
%	\includegraphics[width = \linewidth]{avg-extreme-uniform.png}
%	\caption{Average-case errors incurred on a synthetic dataset with an extreme-valued collection of $\{m_\ell\}$ values and uniform samples; here, the errors are shown on a log-scale}
%	\label{fig:avg-extreme-uniform}
%\end{figure}

\begin{figure}[!t]
	\centering
	\includegraphics[width = \linewidth]{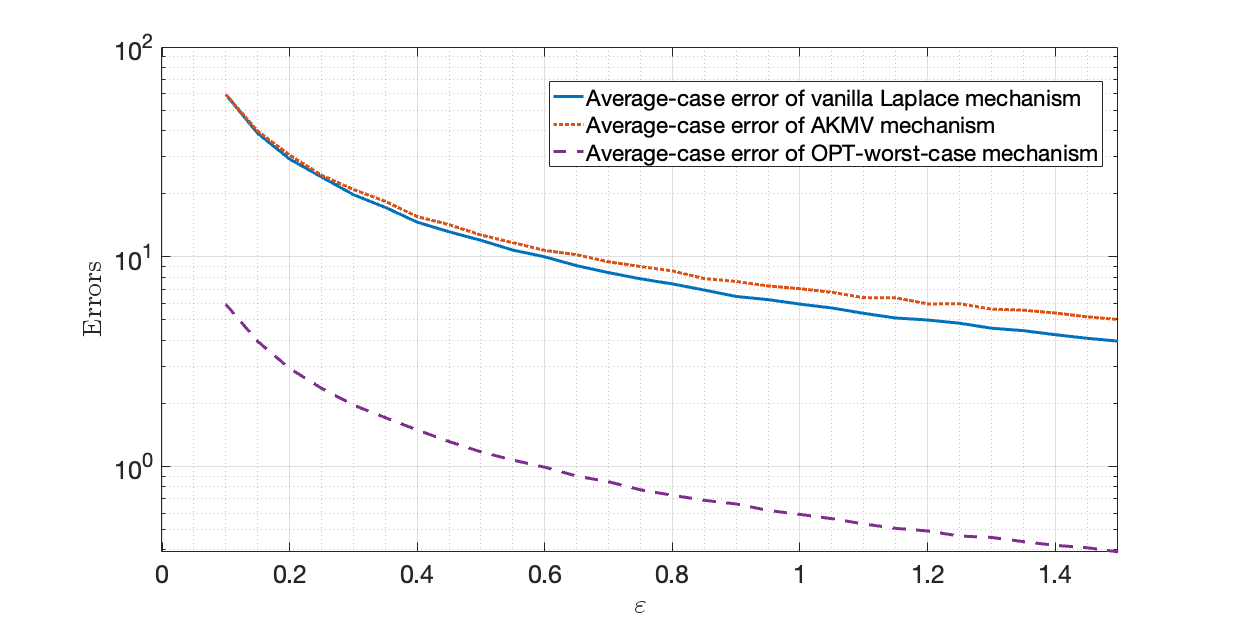}
	\caption{Average-case errors using an extreme-valued collection of $\{m_\ell\}$ values and projected Gaussian samples}
	\label{fig:avg-extreme-normal}
\end{figure}
\section{Conclusion}
\label{sec:conclusion}
In this paper, we revisited the fundamental problem of releasing the sample mean under user-level differential privacy (DP) -- a problem that is well-studied in the DP literature, but typically in the presence of additional (e.g., i.i.d.) assumptions on the distribution of data samples. In our work, we make no distributional assumptions on the dataset; instead, we adopt a worst-case approach to studying the estimation error. We then explicitly characterize this error for a broad class of strategies for bounding user contributions by some combination of clipping the numbers of user contributions and/or clipping the sample values themselves. Our analysis allows us to identify the bounding strategy that is optimal with respect to our worst-case error metric. Via numerical experiments, we demonstrate that our strategy is not only optimal for the worst-case error, but also performs much better than the well-known strategy in \cite{amin} for datasets with i.i.d. samples.

An interesting line of future work will be to extend our worst-case error analysis to the user-level DP release of other statistics, and identify the optimal bounding strategies in those cases as well.

\bibliographystyle{IEEEtran}
{\footnotesize
	\bibliography{references.bib}}

@INPROCEEDINGS{usereg3,
	author={Wang, Zhibo and Song, Mengkai and Zhang, Zhifei and Song, Yang and Wang, Qian and Qi, Hairong},
	booktitle={IEEE INFOCOM 2019 - IEEE Conference on Computer Communications}, 
	title={Beyond Inferring Class Representatives: User-Level Privacy Leakage From Federated Learning}, 
	year={2019},
	volume={},
	number={},
	pages={2512-2520},
	doi={10.1109/INFOCOM.2019.8737416}}

@inproceedings{hetero,
	title = {Mean Estimation with User-level Privacy under Data Heterogeneity},
	booktitle = {NeurIPS},
	author = {Rachel Cummings and Vitaly Feldman and Audra McMillan and Kunal Talwar},
	year = {2022},
	URL = { https://openreview.net/pdf?id=oYbQDV3mon-}
}

@inproceedings{hetero-user-level,
	author = {Cummings, Rachel and Feldman, Vitaly and McMillan, Audra and Talwar, Kunal},
	title = {Mean estimation with user-level privacy under data heterogeneity},
	year = {2024},
	isbn = {9781713871088},
	publisher = {Curran Associates Inc.},
	address = {Red Hook, NY, USA},
	booktitle = {Proceedings of the 36th International Conference on Neural Information Processing Systems},
	articleno = {2113},
	numpages = {13},
	location = {New Orleans, LA, USA},
	series = {NIPS '22}
}

@inproceedings{hist-user-level,
	author = {Liu, Yuhan and Suresh, Ananda Theertha and Zhu, Wennan and Kairouz, Peter and Gruteser, Marco},
	title = {Algorithms for bounding contribution for histogram estimation under user-level privacyAlgorithms for bounding contribution for histogram estimation under user-level privacy},
	year = {2023},
	publisher = {JMLR.org},
	booktitle = {Proceedings of the 40th International Conference on Machine Learning},
	articleno = {911},
	numpages = {28},
	location = {Honolulu, Hawaii, USA},
	series = {ICML'23}
}

@misc{tit-preprint,
	title={Improving the Privacy Loss Under User-Level DP Composition for Fixed Estimation Error}, 
	author={V. Arvind Rameshwar and Anshoo Tandon},
	year={2024},
	eprint={2405.06261},
	archivePrefix={arXiv},
	primaryClass={cs.CR},
	url={https://arxiv.org/abs/2405.06261}, 
}

@article{kairouz-survey,
	url = {http://dx.doi.org/10.1561/2200000083},
	year = {2021},
	volume = {14},
	journal = {Foundations and Trends® in Machine Learning},
	title = {Advances and Open Problems in Federated Learning},
	doi = {10.1561/2200000083},
	issn = {1935-8237},
	number = {1–2},
	pages = {1-210},
	author = {Peter Kairouz and others}
}

@ARTICLE{staircase,
	author={Geng, Quan and Kairouz, Peter and Oh, Sewoong and Viswanath, Pramod},
	journal={IEEE Journal of Selected Topics in Signal Processing}, 
	title={The staircase mechanism in differential privacy}, 
	year={2015},
	volume={9},
	number={7},
	pages={1176-1184},
	doi={10.1109/JSTSP.2015.2425831}}

@INPROCEEDINGS{opt,
	author={Geng, Quan and Viswanath, Pramod},
	booktitle={2014 IEEE International Symposium on Information Theory}, 
	title={The optimal mechanism in differential privacy}, 
	year={2014},
	volume={},
	number={},
	pages={2371-2375},
	doi={10.1109/ISIT.2014.6875258}}

@Inbook{vadhan2017,
	author="Vadhan, Salil",
	editor="Lindell, Yehuda",
	title="The Complexity of Differential Privacy",
	bookTitle="Tutorials on the Foundations of Cryptography: Dedicated to Oded Goldreich",
	year="2017",
	publisher="Springer International Publishing",
	address="Cham",
	pages="347--450",
	isbn="978-3-319-57048-8",
	doi="10.1007/978-3-319-57048-8_7",
	url="https://doi.org/10.1007/978-3-319-57048-8_7"
}

@inproceedings{dpsgd,
	author = {Abadi, Martin and Chu, Andy and Goodfellow, Ian and McMahan, H. Brendan and Mironov, Ilya and Talwar, Kunal and Zhang, Li},
	title = {Deep Learning with Differential Privacy},
	year = {2016},
	isbn = {9781450341394},
	publisher = {Association for Computing Machinery},
	address = {New York, NY, USA},
	url = {https://doi.org/10.1145/2976749.2978318},
	doi = {10.1145/2976749.2978318},
	booktitle = {Proceedings of the 2016 ACM SIGSAC Conference on Computer and Communications Security},
	pages = {308–318},
	numpages = {11},
	location = {Vienna, Austria},
	series = {CCS '16}
}

@inproceedings{usereg4,
	author       = {H. Brendan McMahan and
	Daniel Ramage and
	Kunal Talwar and
	Li Zhang},
	title        = {Learning Differentially Private Recurrent Language Models},
	booktitle    = {6th International Conference on Learning Representations, {ICLR} 2018,
	Vancouver, BC, Canada, April 30 - May 3, 2018, Conference Track Proceedings},
	publisher    = {OpenReview.net},
	year         = {2018},
	url          = {https://openreview.net/forum?id=BJ0hF1Z0b},
	bibsource    = {dblp computer science bibliography, https://dblp.org}
}

@book{cover_thomas, place={New Delhi}, title={Elements of Information Theory}, publisher={Wiley-India}, author={Cover, T. M. and Thomas, Joy A.}, year={2010}, edition={2nd}}

@misc{carath,
 	title={Basic Properties of Convex Sets},
 	author={Jean Gallier},
 	note={lecture notes for CIS 610: Advanced Geometric Methods in Computer Science},
 	url={https://www.cis.upenn.edu/~cis6100/convex1-09.pdf},
 }

@article{dp_spcom, title={Mean Estimation with User-Level Privacy for Spatio-Temporal IoT Datasets}, journal={Submitted to the IEEE International Conference on Signal Processing and Communications (SPCOM)}, author={P. Gupta and V. A. Rameshwar and A. Tandon and N. Chakraborty}, year={2024}}

@book{boyd, place={New Delhi}, title={Convex Optimization}, publisher={ Cambridge: Cambridge University Press}, author={S. Boyd and L. Vandenberghe}, year={2004}}

@ARTICLE{dp_preprint,
	author = {{Arvind Rameshwar}, V. and {Tandon}, Anshoo and {Gupta}, Prajjwal and {Chakraborty}, Novoneel and {Sharma}, Abhay},
	title = "{Mean Estimation with User-Level Privacy for Spatio-Temporal IoT Datasets}",
	journal = {arXiv e-prints},
	year = 2024,
	month = jan,
	eid = {arXiv:2401.15906},
	pages = {arXiv:2401.15906},
	doi = {10.48550/arXiv.2401.15906},
	archivePrefix = {arXiv},
	eprint = {2401.15906},
	primaryClass = {cs.CR},
	adsurl = {https://ui.adsabs.harvard.edu/abs/2024arXiv240115906A}
}

@article{dwork06, title={Calibrating noise to sensitivity in private data analysis}, DOI={10.1007/11681878_14}, journal={Theory of Cryptography}, author={Dwork, Cynthia and McSherry, Frank and Nissim, Kobbi and Smith, Adam}, year={2006}, pages={265–284}}

@inproceedings{
	userlevel,
	title={Learning with User-Level Privacy},
	author={Daniel Asher Nathan Levy and Ziteng Sun and Kareem Amin and Satyen Kale and Alex Kulesza and Mehryar Mohri and Ananda Theertha Suresh},
	booktitle={Advances in Neural Information Processing Systems},
	editor={A. Beygelzimer and Y. Dauphin and P. Liang and J. Wortman Vaughan},
	year={2021},
	url={https://openreview.net/forum?id=G1jmxFOtY_}
}

@article{dworkroth,
	url = {http://dx.doi.org/10.1561/0400000042},
	year = {2014},
	volume = {9},
	journal = {Foundations and Trends® in Theoretical Computer Science},
	title = {The Algorithmic Foundations of Differential Privacy},
	doi = {10.1561/0400000042},
	issn = {1551-305X},
	number = {3–4},
	pages = {211-407},
	author = {Cynthia Dwork and Aaron Roth}
}

@InProceedings{amin,
	title = 	 {Bounding User Contributions: A Bias-Variance Trade-off in Differential Privacy},
	author =       {Amin, Kareem and Kulesza, Alex and Munoz, Andres and Vassilvtiskii, Sergei},
	booktitle = 	 {Proceedings of the 36th International Conference on Machine Learning},
	pages = 	 {263--271},
	year = 	 {2019},
	editor = 	 {Chaudhuri, Kamalika and Salakhutdinov, Ruslan},
	volume = 	 {97},
	series = 	 {Proceedings of Machine Learning Research},
	month = 	 {09--15 Jun},
	publisher =    {PMLR},
	url = 	 {https://proceedings.mlr.press/v97/amin19a.html}
}
\appendix
\section{Characterizing $\ell_1$-Projections Onto $\Delta_\alpha$}
\label{app:projection}
In this section, we argue that the condition in the second statement of Lemma \ref{lem:proj1} is sufficient for the vector $\mathbf{y}$ to be an $\ell_1$-projection of $\mathbf{a}\in \Delta_U\setminus \Delta_\alpha$, for $\alpha\leq U$.

Following Lemma \ref{lem:proj1}, consider the collection $\mathcal{S}\subseteq \mathbb{R}^d$ of points defined as $\mathcal{S}:= \{\mathbf{z}\in \delta_\alpha:\ z_i\leq a_i,\ \text{for all $i\in [d]$}\}$. Note that  when $a_i>\alpha$, for all $i\in [d]$, we have $\mathcal{S} = \delta_\alpha$. Observe that $\mathcal{S}$ is a convex subset of $\mathbb{R}^d$; hence, by a version of Carath\'eodory's Theorem (see the remark after \cite[Thm. 15.3.5]{cover_thomas} and \cite{carath}), any point $\mathbf{z}$ in $\mathcal{S}$ can be written as a convex combination of finitely many, in particular, $d$ points in $\mathcal{S}$. Hence, consider any such collection $\mathcal{Z} = \{\mathbf{z}_1,\ldots,\mathbf{z}_d\}$ whose convex hull equals $\mathcal{S}$. The following claim then holds.
%Now, consider the case where $\mathbf{a}$ is such that $a_i\leq \alpha$, for some $i\in [d]$, and hence that $a_j\geq \alpha$, for some $j\in [d], j\neq i$.  the proof of the sufficiency of the second statement of Lemma \ref{lem:proj1} for this case is similar to that for the first case, and is omitted.

%For the first case above, consider the collection $\mathcal{Z} = \{\mathbf{z}_1,\ldots,\mathbf{z}_d\}$. Here, for $k\neq i$, we define $\mathbf{z}_k = (z_{k,1},\ldots,z_{k,d})$ to be such that $z_{k,i} = a_i$ and $z_{k,k} = \alpha-a_i$,  with $z_{k,r} = 0$, for $r\notin \{i,k\}$. We define $\mathbf{z}_i = (z_{i,1},\ldots,z_{i,d})$ to obey $z_{i,j} = \alpha$, with $z_{i,r} = 0$, for $r\neq j$. 

\begin{proposition}
	\label{prop:proj1}
	For any $\mathbf{z}\in \mathcal{S}$, we have that $\lVert \mathbf{a} - \mathbf{z}\rVert_1$ is a constant.
\end{proposition}
\begin{proof}
	Recall that by the version of Carath\'eodory's Theorem above, any point $\mathbf{z}\in \mathcal{S}$ can be written as $\mathbf{z} = \sum_{k=1}^d \lambda_k \mathbf{z}_k$, where $\lambda_k\geq 0$, with $\sum_{k\leq d} \lambda_k = 1$.  It suffices to show that $\lVert \mathbf{a} - \mathbf{z}\rVert_1$ is independent of $\{\lambda_k\}$. To see this, note that
	\begin{align*}
		\lVert \mathbf{a} - \mathbf{z}\rVert_1&= \sum_{r\leq d} (a_r-z_r)\\
		&=\sum_{r\leq d} (a_r-\sum_{k\leq d} \lambda_k z_{k,r})\\
		&= \sum_{r\leq d} a_r - \sum_{k\leq d} \lambda_k\cdot \sum_{r\leq d} z_{k,r} = \lVert \mathbf{a}\rVert_1 - \alpha,
	\end{align*}
which is independent of $\{\lambda_k\}$.
\end{proof}
We hence have that the condition in Lemma \ref{lem:proj1} is sufficient for the vector $\mathbf{y}$ in its statement to be an $\ell_1$-projection of $\mathbf{a}$. In practice, given a vector $\mathbf{a}\in \Delta_U\setminus \Delta_\alpha$, one can use the vector $\mathbf{y} = \frac{\alpha}{\lVert \mathbf{a}\rVert_1} \mathbf{a}$ as an $\ell_1$-projection.
%\clearpage
%\appendices
%\section{Proof of Proposition \ref{prop:varsens}}
%\input{app-lem-app1.tex}
%\section{Proof of Theorem \ref{thm:varworst}}
%\input{varworst.tex}
%\section{Experimental Setup}
%\input{app-experiments.tex}
\end{document}